\documentclass[11pt]{article}
\usepackage[utf8]{inputenc}
\usepackage[american]{babel}
\usepackage{amssymb}
\usepackage{amsmath}
\usepackage{thmtools}
\usepackage{listings}
\usepackage{graphicx}
\usepackage{color}
\usepackage{tikz}
\usepackage{subfig}
\usepackage{framed}
\usepackage{xspace}
\usepackage{todonotes}
\usepackage{enumerate}
\usepackage{xargs}
\usepackage{nicefrac}
\usepackage{csquotes}
\usepackage{mathtools}  
\usepackage[ruled,noline,linesnumbered]{algorithm2e}
\usepackage[backend=bibtex,doi=false,isbn=false,url=false,firstinits=true,maxcitenames=2,maxbibnames=99]{biblatex}

\AtEveryBibitem{
    \clearlist{location}
    \clearfield{pages}
    \clearfield{url}
    \clearname{editor}
}

\usetikzlibrary{decorations}
\usetikzlibrary{decorations.pathreplacing}
\usetikzlibrary{patterns}
\usetikzlibrary{arrows}

\newtheorem{theorem}{Theorem}[section]
\newtheorem{lemma}[theorem]{Lemma}
\newtheorem{corollary}[theorem]{Corollary}

\newtheorem{proposition}[theorem]{Proposition}

\newenvironment{proof}{\noindent {\bf Proof.}\ }{\qed\par\vskip 4mm\par}

\newcommand{\qed}{\hfill $\square$}

\newcommandx*{\LDAUOmicron}[2][1=@pkling_false]{\mathcal{O}\ifthenelse{\equal{#1}{small}}{\bigl(#2\bigr)}{\left(#2\right)}}
\newcommandx*{\LDAUomicron}[2][1=@pkling_false]{\mathrm{o}\ifthenelse{\equal{#1}{small}}{\bigl(#2\bigr)}{\left(#2\right)}}
\newcommandx*{\LDAUOmega}[2][1=@pkling_false]{\Omega\ifthenelse{\equal{#1}{small}}{\bigl(#2\bigr)}{\left(#2\right)}}
\newcommandx*{\LDAUomega}[2][1=@pkling_false]{\omega\ifthenelse{\equal{#1}{small}}{\bigl(#2\bigr)}{\left(#2\right)}}
\newcommandx*{\LDAUTheta}[2][1=@pkling_false]{\Theta\ifthenelse{\equal{#1}{small}}{\bigl(#2\bigr)}{\left(#2\right)}}

\newcommandx*{\set}[2][2=@pkling_false]{\left\{#1\ifthenelse{\equal{#2}{@pkling_false}}{}{\;\middle|\;#2}\right\}}

\DeclareMathOperator{\PoA}{PoA}
\DeclareMathOperator{\diam}{diam}
\DeclareMathOperator{\girth}{girth}
\DeclareMathOperator{\OPT}{OPT}

\makeatletter
\xspaceaddexceptions{\check@icr}
\makeatother

\title{%
    Multilevel Network Games%
    \thanks{This work was partially supported by the German Research Foundation (DFG) within the Collaborative Research Center ``On-The-Fly Computing'' (SFB 901), by the EU within FET project MULTIPLEX under contract no.\ 317532, and the International Graduate School ``Dynamic Intelligent Systems''.}%
    ~~\thanks{An abstract of this paper has been accepted for publication in the proceedings of the 7th International Symposium on Algorithmic Game Theory (SAGT), available at www.springerlink.com \cite{2014sagtmultilevel}. An extended abstract of this paper has been accepted for publication in the proceedings of the 10th International Conference on Web and Internet Economics (WINE), available at www.springerlink.com \cite{ACJS14}.}
}

\author{
    Sebastian Abshoff
    \and
    Andreas Cord-Landwehr
    \and
    Daniel Jung
    \and
    Alexander Skopalik
    \newline
}

\date{
    Heinz Nixdorf Institute \& Computer Science Department\\[0.2em]
    University of Paderborn (Germany)\\[0.2em]
    Fürstenallee 11, 33102 Paderborn\\[0.2em]
}

\begin{document}

\maketitle

\thispagestyle{empty}

\begin{abstract}
We consider a multilevel network game, where nodes can improve their communication costs by connecting to a high-speed network.
The $n$ nodes are connected by a static network and each node can decide individually to become a gateway to the high-speed network.
The goal of a node $v$ is to minimize its private costs, i.e., the sum (SUM-game) or maximum (MAX-game) of communication distances from $v$ to all other nodes plus a fixed price $\alpha > 0$ if it decides to be a gateway.
Between gateways the communication distance is $0$, and gateways also improve other nodes' distances by behaving as shortcuts.
For the SUM-game, we show that for $\alpha \leq n-1$, the price of anarchy is $\LDAUTheta{n/\sqrt{\alpha}}$ and in this range equilibria always exist.
In range $\alpha \in (n-1,n(n-1))$ the price of anarchy is $\LDAUTheta{\sqrt{\alpha}}$, and for $\alpha \geq n(n-1)$ it is constant.
For the MAX-game, we show that the price of anarchy is either $\LDAUTheta{1 + n/\sqrt{\alpha}}$, for $\alpha\geq 1$, or else $1$.
Given a graph with girth of at least $4\alpha$, equilibria always exist.
Concerning the dynamics, both the SUM-game and the MAX-game are not potential games.
For the SUM-game, we even show that it is not weakly acyclic.
\end{abstract}

\section{Introduction}
Today's networks, like the Internet, do not consist of one but a mixture of several interconnected networks.
Every network has individual qualities and hence the total performance of the network becomes a mixture of these individual properties.
Typically, one can categorize those different networks into high-speed backbone networks and low-speed general purpose networks.
Given the fact that nodes in Internet-like networks establish their connections in an uncoordinated and selfish way, it becomes a challenging question to understand the evolution and outcome of those networks.

We model and analyze the interaction of two networks: a low speed general purpose network and a high-speed backbone network.
Every node can decide individually if it wants to connect to the high-speed network for a fixed price $\alpha$ in order to minimize its private costs, i.e., the costs of connecting to the high-speed network plus the costs for communicating with other nodes.
The communication costs of a node are given by the sum or maximum distance to all other nodes in the network, possibly improved by shortcuts through the high-speed network.
Having two nodes that are both connected to the high-speed network, they provide a shortcut of a fixed (very small) edge length.
In our model, we assume the shortcut edge length to be less than $1$ divided by the number of nodes and normalize it to be $0$.

\paragraph{Model and Notations.}
We consider a set $V$ of $n$ nodes forming an undirected connected graph $G \coloneqq (V,E)$.
Each node of this graph can connect to a high-speed network by paying a fixed price $\alpha > 0$.
A node connected to the high-speed network is called a \emph{gateway} and we assume communication distances between each pair of gateway nodes to be $0$.
The shortest path distance between two nodes $u,v$ in $G$ is given by $d(u,v)$, whereas we consider $d(u,v)$ to be the hop distance.
Having a set of gateways $S$, we define the communication distance $\delta(u,v) \coloneqq \min\{d(u,v), d(u,S) + d(S,v)\}$.
Each node $v\in V$ aims to minimize its private costs by selfishly deciding whether to connect to the high-speed network.
We identify the set of gateways $S$ with the current strategy profile, i.e., nodes in $S$ are gateways and nodes in $V\setminus S$ are non-gateways.

The private costs of a node in the \emph{SUM-game} are $c_v(S) \coloneqq |S\cap\{v\}|\alpha + \sum_{u\in V} \delta(v,u)$.
For the \emph{MAX-game}, the private cost function is $c_v(S) \coloneqq |S\cap\{v\}|\alpha + \max_{u\in V} \delta(v,u)$.
For both games, the social costs are $c(S) \coloneqq \sum_{v\in V} c_v(S)$.

If a node improves its private costs by changing its strategy from non-gateway to gateway or vice versa, we call this an \emph{improving response} (IR).
For an IR where a node $v$ changes its strategy to be a gateway, we say that $v$ \emph{opens}.
Analogously, we say $v$ \emph{closes} if it changes its strategy from gateway to non-gateway.
We call a strategy profile $S$ a (pure) Nash equilibrium (NE) if no node can perform an IR.
We require that there is always at least one gateway in the graph, i.e., the last gateway is not allowed to close even if that strategy change is an IR.

\paragraph{Our Questions.}
A main objective of the research on network games is the analysis of the \emph{price of stability} and the \emph{price of anarchy}.
Both measure the quality of Nash equilibria by comparing their social costs to the smallest social costs possible for a given graph and $\alpha$ value.
The price of stability, see for example \cite{anshelevich2003,anshelevich2004}, is defined as the ratio of the smallest social costs of any Nash equilibrium and the optimal social costs.
The price of anarchy (PoA), introduced in \cite{koustoupias1999}, is defined as the ratio of the biggest social costs of any Nash equilibrium and the optimal social costs.

We analyze the convergence processes of the games by questioning whether they provide the \emph{finite improvement property} or (lesser) whether they are \emph{weakly acyclic}.
A game with the finite improvement property (FIPG) guarantees that, when starting from any initial state, \emph{every} sequence of IRs eventually converges to a NE state, i.e., every sequence of IRs is finite.
\textcite{monderer1996} showed that a game is a FIPG if and only if there exists a generalized ordinal potential function $\Phi: V \rightarrow \mathbb{R}$ that maps strategy profiles to real numbers such that if a node performs an IR the potential value decreases.
A game is called \emph{weakly acyclic} (WAG) \cite{young1993evolution} if, starting from any initial strategy profile, there exists \emph{some} finite sequence of IRs that eventually converges to a NE state.

\paragraph{Related Work.}
Network creation games (NCG) are an established model to study the evolution and quality of networks established by selfishly acting nodes.
In these games, nodes can decide individually which edges they want to buy (each for a fixed price $\alpha > 0$) in order to minimize their private costs.
For a node, the private costs are either the sum (SUM-game, \textcite{fabrikant2003}) or maximum (MAX-game, \textcite{demaine2007}) of the distances to all other nodes in the network plus the costs of the bought edges.

The task of describing the maximal possible loss by selfish behavior was formalized as the \emph{price of anarchy} and first discussed by \textcite{fabrikant2003} for the SUM-game.
The authors proved an upper bound of $\LDAUOmicron{\sqrt{\alpha}}$ on the price of anarchy (PoA) in the case of $\alpha<n^2$, and a constant PoA otherwise.
Later, \textcite{albers2006} proved a constant PoA for $\alpha=\LDAUOmicron{\sqrt{n}}$ and the first sublinear worst case bound of $\LDAUOmicron{n^{1/3}}$ for general $\alpha$.
\textcite{demaine2007} were the first to prove an $\LDAUOmicron{n^{\varepsilon}}$ bound for $\alpha$ in the range of $\LDAUOmega{n}$ and $\LDAUomicron{n\lg n}$.
Recently, by \textcite{mihalak2010} and improved by \cite{mihalak2013treeEquilibria}, it was shown that for $\alpha \geq 65 n$ all equilibria are trees (and thus the PoA is constant).
For non-integral constant values of $\alpha > 2$, \textcite{hamilton2013anarchyIsFree} showed that the PoA tends to $1$ as $n\rightarrow\infty$.

For the MAX-game, \textcite{demaine2007} showed that the PoA is at most $2$ for $\alpha\geq n$, for $\alpha$ in range for $2\sqrt{\lg n}\leq \alpha\leq n$ it is $\LDAUOmicron{\min\{4^{\sqrt{\lg n}},(n/\alpha)^{1/3}\}}$, and $\LDAUOmicron{n^{2/\alpha}}$ for $\alpha < 2\sqrt{\lg n}$.
For $\alpha>129$, \textcite{mihalak2010} showed, like in the average distance version, that all equilibria are trees and the PoA is constant.

In several subsequent papers, different approaches were done to simplify the games and also to enable nodes to compute their best responses in polynomial time (which is not possible in the original games).
\textcite{alon2010} introduced the \emph{basic network creation game}, where the operation of a node only consists of swapping some of its incident edges, i.e., redirecting them to other nodes.
There is no dependence on a cost parameter $\alpha$ and best-responses can be computed in polynomial time.
Restricting the initial network to trees, the only equilibrium in the SUM-game is a star graph.
Without restrictions, all (swap) equilibria are proven to have a diameter of $2^{\LDAUOmicron{\sqrt{\log n}}}$, which is also the PoA.
For the MAX-version, the authors provide an equilibrium network with diameter $\LDAUTheta{\sqrt{n}}$.

In \cite{lenzner2012greedy}, Lenzner introduced a different approach by taking the original NCGs from \cite{fabrikant2003}, but restricting the operation of a node to single edge changes.
This model allows for polynomial time computable best responses, but at the same time its equilibria give a $3$-approximation to the equilibria of the original game.

For a variety of these games, \textcite{lenzner2013} studied convergence properties, namely, whether those games always converge for any sequence of IRs (FIPG) or whether at least there exists a finite sequence of IRs from any initial state such that the game eventually converges to a NE (WAG).
For all researched games, with the only exception of basic network creation games on initial tree networks, they provided negative convergence results.

\paragraph{Our Results.}
We introduce a new network game that focuses on dynamics in multilevel networks.
In particular, we study how nodes of a general purpose network interact with a high-speed network when deciding selfishly whether to connect or not connect to the high-speed network.
For both the SUM- and the MAX-game it is NP-hard to find an optimal placement of gateways (Theorem~\ref{thm:NPhardness} and \ref{thm:NPhardnessMax}).

For the SUM-game, we show that for $\alpha \leq n-1$ and $\alpha > n(n-1)$ equilibria always exist and that the PoA is $\LDAUTheta{1 + n/\sqrt{\alpha}}$ (Theorem~\ref{thm:sumPoAresults}).
For the range $\alpha\in (n-1,n(n-1))$ we upper bound the PoA by $\LDAUOmicron{\sqrt{\alpha}}$.
Concerning the dynamics, the SUM-game is no FIPG for wide ranges of parameter values $\alpha$ (Theorem~\ref{thm:sumNotFIPGgeneral}).
And we further show that it is not even a WAG (\ref{thm:SumNotWAG}).
Yet, we can provide convergence properties for special cases.

For the MAX-game, we show that equilibria always exist if the girth is at least $4\alpha$ (which is always true for trees).
Like in the SUM-game, the PoA is $\LDAUTheta{1 + n/\sqrt{\alpha}}$ for $\alpha \geq 1$ (Theorem~\ref{thm:maxPoaResults}), and otherwise $1$.
The MAX-game is also not a FIPG (Theorem~\ref{thm:maxNotFIPG}).

\section{The SUM-Game}
In this section, we consider the equilibria in the SUM-game.
First we show that in general it is NP-hard to compute socially optimal placements of gateways.
For a graph $G$ with $\alpha \leq n-1$ or $n\diam(G) < \alpha$ we can show that equilibria always exist and provide tight price of anarchy results for these ranges.

\begin{theorem}
\label{thm:NPhardness}
Given a graph $G=(V,E)$ it is NP-hard to compute an optimal set of gateways $S\subseteq V$ that minimizes the social costs in the SUM-game.
\end{theorem}
\begin{proof}
For $n, m > 4$, let $X\coloneqq\{x_1,\ldots, x_m\}$ be a set of elements and $S_1,\ldots,S_n\subseteq X$ sets that form an instance of the NP-complete Set-Cover problem (cf.\ \textcite{karp1972}).
Given the Set-Cover instance, we construct an instance of the SUM-game as follows (cf.\ Figure~\ref{fig:np-hardness-construction}).
First, we create a clique $C$ of $k$ nodes and mark one of its nodes as $c$.
For every set $S_i$, we create a corresponding node $S_i$ and connect each set node to $c$.
For every element $x_i\in X$, we create $w$-many nodes $x_i^1,\ldots,x_i^w$ and connect all $x_i^j,i=1,\ldots,m, j=1,\ldots,w$ to the set nodes $S_l$ with $x_i\in S_l$.
The parameters are $w \coloneqq n$, $k \coloneqq m-1$, and $\alpha \coloneqq 4n(m-1)$.

For now, consider that $c$ is a gateway node in the optimal solution $S^{\OPT}$ (we will prove this claim later).
We claim that then no other node $v\in C$, $v \not= c$ is a gateway.
For this, assume that $l$ further clique nodes are open and compute the social costs decrease by closing all clique nodes except of $c$.
The decrease is at least $l\alpha - 2l(wm + n) - l(l+1) > 0$ and hence $c$ is the only node in $C\cap S^{\OPT}$.

Next, for an element $x_i$, consider the corresponding element nodes $x_i^1,\ldots,x_i^w$ and a set $S_j$ such that $x_i\in S_j$.
If there is any $x_i^l\in S$ and $S_j \not\in S$, closing $x_i^l$ and opening $S_j$ does not increase the social costs.
Hence, we can assume that in $S^{\OPT}$ there is no closed set node with an open element node.
Now, let $S_j$ be an open set node and assume that for $x_i\in S_j$, there are $l$-many open element nodes $x_i^1,\ldots,x_i^l$ (w.l.o.g.\ we consider the first $l$).
Closing all of these element nodes reduces the social costs by at least
$l\alpha - 2l(k + n + 2(l-1) + (w-l) + (m-1)w) = l\alpha - 2l(wm + n + k + l - 2) > 0$.
Now, given the closed element nodes $x_i^1,\ldots,x_i^w$ such that for all $S_j$, with $x_i\in S_j$, the set nodes are closed, opening $S_j$ reduces the social costs by at least
$2(kw + (m-1)w + (n-1)) - \alpha > 0$.
Contrary, opening a set node whose element nodes are already completely covered increases the social costs by at least
$\alpha - 2(k + mw + n-1) > 0$.

Finally, we can see that $c$ actually has to be a gateway in $S^{\OPT}$.
For this, consider an arbitrary optimal setting with all clique nodes closed (if one clique node is open, we can close it and open $c$ without increasing the social costs).
When opening $c$, we know that without increasing social costs we can close all element nodes and open corresponding set nodes.
Hence, when opening $c$ we can assume that all element nodes are closed and that for each element node a corresponding set is open.
This gives a social costs decrease by opening $c$ of at least $2kmw - \alpha > 0$.

Hence, the socially optimal solution $S^{\OPT}$ is given by a gateway node $c$ and a minimal number of set nodes such that all element nodes are covered.
\end{proof}

\begin{figure}[t]
    \centering
%
\tikzstyle{open}=[circle,fill=blue!50,draw=black,minimum size=6pt,inner sep=0pt]
\tikzstyle{peer}=[circle,fill=white!25,draw=black,minimum size=6pt,inner sep=0pt]
\begin{tikzpicture}
%
\draw (0,3.7) circle (1);
\node[open](c) [label=left:$c$] at (0,3) {};
\node[peer](c1) at (-0.7,3.5) {};
\node[peer](c2) at (-0.4,4.3) {};
\node[peer](c3) at (0.4,4.3) {};
\node[peer](c4) at (0.7,3.5) {};
\draw (c) -- (c1) -- (c2) -- (c3) -- (c4) -- (c);
\draw (c) -- (c2) -- (c4) -- (c1) -- (c3) -- (c);
\node at (2,3.7) {clique $C$};

\node[peer] (s1) at (-1.5,2) {};
\node[open] (s2) at (-0.75,2) {};
\node[peer] (s3) at (0,2) {};
\node[open] (s4) at (1.5,2) {};
\node (sDots) at (0.75,2) {$\cdots$};
\draw (c) -- (s1);
\draw (c) -- (s2);
\draw (c) -- (s3);
\draw (c) -- (s4);
\node at (3.2,2.05) {sets $S_1,\dots,S_n$};

\node[peer] (x1a) at (-2,1) {};
\node[peer] (x1b) at (-1.9,1) {};
\node[peer] (x1c) at (-1.8,1) {};

\node[peer] (x2a) at (-1,1) {};
\node[peer] (x2b) at (-0.9,1) {};
\node[peer] (x2c) at (-0.8,1) {};

\node[peer] (x3a) at (0,1) {};
\node[peer] (x3b) at (0.1,1) {};
\node[peer] (x3c) at (0.2,1) {};

\node (xDots) at (0.75,1) {$\cdots$};

\node[peer] (x4c) at (1.8,1) {};
\node[peer] (x4b) at (1.9,1) {};
\node[peer] (x4a) at (2,1) {};

\draw (s1) -- (x1b);
\draw (s1) -- (x3b);
\draw (s2) -- (x2b);
\draw (s2) -- (x1b);
\draw (s3) -- (x4b);
\draw (s3) -- (x1b);
\draw (s4) -- (x3b);
\draw (s4) -- (x4b);

\node at (4.5,1.05) {elements $x_1^{(\cdot)},\dots,x_m^{(\cdot)}$};

\end{tikzpicture}
    \caption{NP-hardness reduction from Set-Cover to optimal gateway placement.}
    \label{fig:np-hardness-construction}
\end{figure}
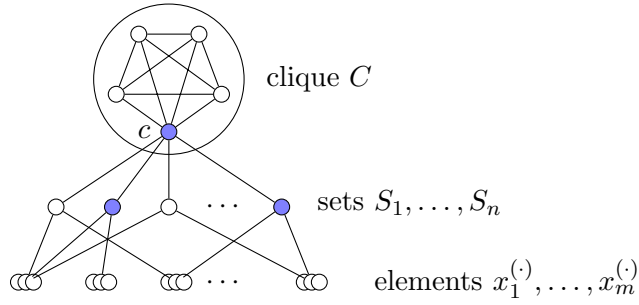

\begin{proposition}
\label{prop:sumNeBigSmallAlphaExistence}
Given a network $G=(V,E)$ and $\alpha \leq n-1$ or $\alpha > n \diam(G)$, then a Nash equilibrium always exists.
\end{proposition}
\begin{proof}
For $\alpha \leq n-1$, consider the strategy profile $S \coloneqq V$, i.e., every node has private costs of $\alpha$.
If any node closes, its distance costs would become at least $n-1$.
Hence, $S=V$ forms a NE.

For $\alpha > n\cdot\diam(G)$, we open an arbitrary node.
Assuming a second node would open, its distance costs decrease by not more than $ n\cdot\diam(G) < \alpha$.
\end{proof}

\begin{lemma}
\label{lemma:socialOptSmallN}
Given a network $G=(V,E)$, $n \coloneqq |V|$, and $\alpha \leq n-1$, then $S = V$ minimizes the social costs in the SUM-game.
\end{lemma}
\begin{proof}
Let $S$ be a socially optimal solution and assume that there are $m$ closed nodes.
When opening all of them, for $v\in S$, the distances to all $m$ nodes reduce by at least $1$ each, while for $u\in V\setminus S$ the distances reduce by at least $n-1$ each.
Hence, changing the strategy profile to $S=V$ changes the social costs by $m\alpha - ((n-m) + m(n-1)) < 0$.
\end{proof}

Note that Lemma~\ref{lemma:socialOptSmallN} does not contradict the NP-hardness proof of Theorem~\ref{thm:NPhardness}, since in that proof $\alpha$ was chosen to be bigger than the number of nodes.

\begin{corollary}
In the SUM-game, for $\alpha \leq n-1$, the price of stability is one.
\end{corollary}
\begin{proof}
The socially optimal solution $S=V$ from Lemma~\ref{lemma:socialOptSmallN} is a NE.
\end{proof}

\subsection{Convergence Properties}
First, we consider some special cases where we can guarantee the existence of convergence sequences in the SUM-game.
Then, we show for the general case that the game is neither a FIPG nor a WAG.
The first proposition directly follows by using the arguments from Proposition~\ref{prop:sumNeBigSmallAlphaExistence}.

\begin{proposition}
Let $G=(V,E)$ be a graph and $S \subseteq V$ an initial set of gateways, then for $\alpha < 1$ or $\alpha > n\cdot\diam(G)$ the SUM-game is a FIPG.
\end{proposition}
\begin{proof}
If $\alpha < 1$, then for every non-gateway node it is an improving response to open.
Also no gateway node will deviate from its strategy and close.
Hence, after at most $n$ improving responses, the strategy profile is $S = V$.
Otherwise, if $\alpha > n(n-1)$, no closed node will open and for every open node it is an improving response to close.
\end{proof}

\begin{proposition}
Let $G=(V,E)$ be a graph and $S$ an initial set of gateways.
If for the connected components $C_1,\ldots,C_k$ of $G\setminus S$ it holds $n - \max_{i=1,\ldots,k}\{|C_i|\} \geq \alpha$, then the SUM-subgame by invoking only nodes $V\setminus S$ is a FIPG.
\end{proposition}
\begin{proof}
First note that if $|S| \geq \alpha$ holds, then for every non-gateway it is an IR to open and with $S = V$ no gateway wants to close.
Now, we look at the nodes in $V\setminus S$ and consider an arbitrary node $v\in C_i$.
For $v$, the shortest paths to more than $\alpha$ many nodes contain gateway nodes.
Hence, becoming a gateway reduces $v$'s distance term by more than $\alpha$.
By opening further nodes, this property will not be harmed and we eventually reach the state $S = V$.
\end{proof}

\begin{proposition}
Let $G=(V,E)$ be a graph and $S$ an initial set of gateways with $|S|=1$.
If $\diam(G) > 2 \alpha + 1$ and $4 \leq \alpha \leq n-1$, then in the SUM-game there exists a sequence of IRs such that the game converges to a NE.
\end{proposition}
\begin{proof}
Let $x\in S$ be the initial gateway and consider $v,u$ being nodes with $d(v,u) > 2 \alpha + 1$.
One of these nodes (say $v$) must have distance greater than $\alpha + 1$ to $x$.
By opening, $v$ reduces its distances to at least half of the nodes on the shortest path to $x$, i.e. by
$\sum_{i=1}^{\lceil \alpha/2 \rceil} (2i-1)
    = \lceil \alpha/2 \rceil (\lceil \alpha/2 \rceil + 1) - \lceil \alpha/2 \rceil
    > \alpha$.

Next, with $S=\{x,v\}$, also $u$ opens, since opening reduces its distances to at least half of the nodes on a shortest path from $u$ to $v$, i.e., $\lceil\alpha\rceil$ many nodes.
Considering the nodes on the shortest path from $u$ to $v$, for each of them it is an IR to open (since opening improves distance to at least $\lceil\alpha\rceil$ many nodes).
Hence, starting from one end of the path, we open them iteratively.
Finally, with $|S| > \alpha$, also all other nodes open and we reach $S = V$, which is a NE.
\end{proof}

In general it is not always possible to find IR sequences such that the SUM-game converges to a NE.
First, we show for $\alpha \in (4,n-1)$ and then for $\alpha\in (\frac{3}{32}n^2 + n , \frac{5}{32}n^2)$ that infinite IR cycles exist, i.e., the game is not a FIPG.

\begin{proposition}
\label{proposition:sumNotFIPGsmallAlpha}
Given $n\in \mathbb{N}$, $n > 7$ and $\alpha\in (4, n-1)$, in general the SUM-game is not a FIPG.
\end{proposition}
\begin{proof}
We construct a graph as given by Figure~\ref{fig:graph-br-cycle} (with $c=1$), i.e., we create a path $u - v - w$, connect further $r$ many nodes to $w$, as well as connect further $n-r-3$ many nodes to $u$.
(The value for $r$ is computed below.)
Starting with only $w$ being a gateway, we specify the constraints under which $u$ and $v$ form an improving response cycle:
\begin{enumerate}[I:]
    \item $u$ opens if $\alpha < 2r + 2$
    \item $v$ opens if $(n-3-r) + r + 2 = n - 1> \alpha$
    \item $u$ closes if $\alpha > r + 2$
    \item $v$ closes if $\alpha > r + 1$
\end{enumerate}
Combining these conditions, we get $r+2 < \alpha < \min\{n-1,2r + 2\}$.
For $2 \leq r \leq n-3$, the interval $(r+2,\min\{n-1,2r + 2\})$ is non-empty and thus for $4 < \alpha < n-1$ the game provides an infinite improving response cycle.
\end{proof}

\begin{figure}[t]
    \centering
\usetikzlibrary{patterns,arrows,decorations.pathreplacing}
\tikzstyle{open}=[circle,fill=blue!50,draw=black,minimum size=6pt,inner sep=0pt]
\tikzstyle{peer}=[circle,fill=white!25,draw=black,minimum size=6pt,inner sep=0pt]
\begin{tikzpicture}
\node[peer](u) at (-4,0) [label=below:$u$] {};
\node[peer](u1) at (-3,0) {};
\node(u2) at (-2.5,0) {};
\node(u3) at (-1.5,0) {};
\node[peer](u4) at (-1,0) {};
\node[peer](w1) at (3,0) {};
\node(w2) at (2.5,0) {};
\node(w3) at (1.5,0) {};
\node[peer](w4) at (1,0) {};
\node[peer](v) at (0,0) [label=below:$v$] {};
\node[open](w) at (4,0) [label=below:$w$] {};
\node[peer](x1) at (-5.5,1) {};
\node[peer](xm) at (-5.5,-1) {};
\node[peer](y1) at (5.5,1) {};
\node[peer](ym) at (5.5,-1) {};
\node at (5.5,0) {$\vdots$};
\node at (-5.5,0) {\vdots};
\node at (-2,0) {$\cdots$};
\node at (2,0) {$\cdots$};
\draw (u) -- (u1) -- (u2);
\draw (u3) -- (u4) -- (v);
\draw (v) -- (w4) -- (w3);
\draw (w2) -- (w1) -- (w);
\draw (x1) -- (u) -- (xm);
\draw (y1) -- (w) -- (ym);
\draw[draw=black](-3.5,0)++(145:1.5) arc (145:215:1.5);
\draw[draw=black](3.5,0)++(145:-1.5) arc (145:215:-1.5);
\node at (-4.5,-1.5) {$n-2c-r-1$ nodes};
\node at (4.5,-1.5) {$r$ nodes};
\draw [decorate,decoration={brace,amplitude=10pt}](-0.5,-0.25) -- (-3.5,-0.25) node[black,midway,yshift=-0.6cm] {$c-1$ nodes};
\draw [decorate,decoration={brace,amplitude=10pt}](3.5,-0.25) -- (0.5,-0.25) node[black,midway,yshift=-0.6cm] {$c-1$ nodes};
\end{tikzpicture}
    \caption{Best response cycle where $u$ and $v$ perform improving responses in turn.}
    \label{fig:graph-br-cycle}
\end{figure}
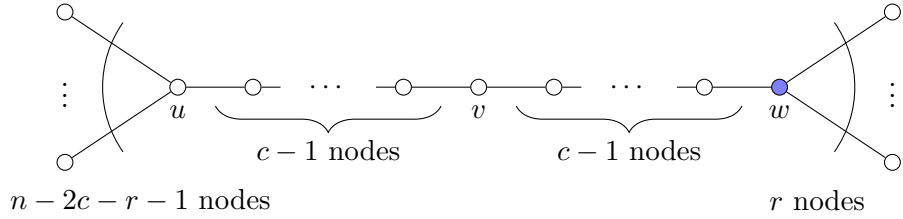

\begin{theorem}
\label{thm:sumNotFIPGgeneral}
Given $n\in \mathbb{N}$, $n > 16$ and $\alpha\in (\frac{3}{32}n^2 + n , \frac{5}{32}n^2)$, in general, the SUM-game is not a FIPG.
\end{theorem}
\begin{proof}
Consider the graph from Figure~\ref{fig:graph-br-cycle}, i.e., a path $u - \ldots - v - \ldots - w$ with $c-1$ many nodes between $u$ and $v$ as well as between $v$ and $w$.
We connect $r$ many nodes to $w$ and $n - 2c - r -1$ many nodes to $u$.
(The range for $r$ will be defined below, $c \coloneqq n/4$.)
Initially, only $w$ is a gateway.
Under the following constraints, $u$ and $v$ form an IR cycle:
\begin{enumerate}[I:]
    \item $u$ opens if $\alpha < \sum_{i=1}^c 2i +2rc$
    \item $v$ opens if $\alpha < 2\sum_{i=1}^{\lfloor c/2 \rfloor} 2i+(n-2c-1)c$
    \item $u$ closes if $\alpha > \sum_{i=1}^{\lfloor c/2 \rfloor} 2i+(r+c+1)c$
    \item $v$ closes if $\alpha > \sum_{i=1}^{\lfloor c/2 \rfloor} 2i+(r+1)c$
\end{enumerate}
To simplify calculations, we assume $4 \mid n$ in the remainder.
Since constraint III implies constraint IV, we only have to consider:
\begin{align}
    \alpha & < c^2 + (2r+1)c\\
    \alpha & < -\frac{3}{2}c^2 + nc\\
    \alpha & > \frac{5}{4}c^2 + (r + 3/2)c
\end{align}
Combining (1) and (3) gives
$r \in \left( \frac{1}{2}\left(\frac{\alpha}{c} - c - 1\right), \frac{\alpha}{c} - \frac{5}{4}c - 3/2 \right)$
as valid range for $r$.
Plugging in $c=n/4$ gives
$r\in (2\alpha/n - n/8 - 1/2, 4\alpha/n - 5n/16 - 3/2)$, i.e., the interval of valid values for $r$ has length of $2\alpha/n - 3n/16 - 1$.
To ensure that there exist integral solutions for $r$, we claim the interval length to be at least $1$, i.e., $2\alpha/n - 3n/16 - 1 \geq 1$, which gives $\alpha \geq n + 3n^2/32$.
Considering (2), which is $\alpha \leq \frac{5}{32}n^2$, we get
$\alpha\in (\frac{3}{32}n^2 + n , \frac{5}{32}n^2)$ as range.
For $n > 16$, this interval is non-empty, and $u$ and $v$ form an infinite IR cycle.
\end{proof}

\begin{theorem}
\label{thm:SumNotWAG}
The SUM-game is not a weakly acyclic game in general.
\end{theorem}
\begin{proof}
For $\alpha \coloneqq 7$ we consider the graph as depicted in Figure~\ref{fig:graph-not-weak-acyclic}.
The graph consists of three nodes $u$, $v$, and $w$ that are connected as a line, a clique $X$ of $\lceil \alpha/2 \rceil$ nodes, a clique $Y$ of $\lfloor \alpha/2 \rfloor$ nodes, and a center node $c$.
All nodes of $X$ are connected to $c$ and to $u$, all nodes of $Y$ are connected to $c$ and to $w$, and further $c$ is connected to $v$.

We consider the initial strategy profile $S=\{w\}$ and argue that there exists a unique sequence of improving responses, such that $u$ and $v$ change their strategy in turn.
Table~\ref{table:nonWagIrs} states that there is always exactly one of these two nodes, which can improve its private costs.
(Note that we explicitly use $\alpha = 7$.)
\end{proof}

\begin{table}
\begin{enumerate}
\item node $u$ opens:\medskip\par
\begin{tabular*}{0.7\textwidth}{c|c|c|c}
    & \textbf{Cost if opened} & \textbf{Cost if closed} & \textbf{State after} \\
    \hline
    $x\in X$ & $2\alpha+2$ & $\alpha+\lfloor \frac{\alpha}{2} \rfloor+6$ & closed \\
    \hline
    $y\in Y$ & $2\alpha+\lceil \frac{\alpha}{2} \rceil+1$ & $\alpha+\lceil \frac{\alpha}{2} \rceil+6$ & closed\\
    \hline
    $u$ & $2\alpha+1$ & $\alpha+2\lfloor \frac{\alpha}{2} \rfloor+5$ & \textbf{opening} \\
    \hline
    $v$ & $2\alpha+\lceil \frac{\alpha}{2} \rceil+2$ & $2\alpha+3$ & closed\\
    \hline
    $w$ & $2\alpha+2\lceil \frac{\alpha}{2} \rceil+5$ & $\alpha+2\lceil \frac{\alpha}{2} \rceil+5$ & opened\\
    \hline
    $c$ & $2\alpha+5$ & $\alpha+5$ & closed\\
    \hline
\end{tabular*}
\medskip

\item node $v$ opens:\medskip\par
\begin{tabular*}{0.7\textwidth}{c|c|c|c}
    & \textbf{Cost if opened} & \textbf{Cost if closed} & \textbf{State after} \\
    \hline
    $x\in X$ & $2\alpha+1$ & $\alpha+\lfloor \frac{\alpha}{2} \rfloor+1$ & closed \\
    \hline
    $y\in Y$ & $2\alpha+1$ & $\alpha+\lceil \frac{\alpha}{2} \rceil+4$ & closed\\
    \hline
    $u$ & $2\alpha+1$ & $\alpha+2\lfloor \frac{\alpha}{2} \rfloor+5$ & opened \\
    \hline
    $v$ & $2\alpha+1$ & $2\alpha+3$ & \textbf{opening}\\
    \hline
    $w$ & $2\alpha+3$ & $\alpha+2\lceil \frac{\alpha}{2} \rceil+5$ & opened\\
    \hline
    $c$ & $2\alpha+1$ & $2\alpha+5$ & closed\\
    \hline
\end{tabular*}
\medskip

\item node $u$ closes:\medskip\par
\begin{tabular*}{0.7\textwidth}{c|c|c|c}
    & \textbf{Cost if opened} & \textbf{Cost if closed} & \textbf{State after} \\
    \hline
    $x\in X$ & $2\alpha$ & $\alpha+3$ & closed \\
    \hline
    $y\in Y$ & $2\alpha$ & $\alpha+3$ & closed\\
    \hline
    $u$ & $2\alpha+1$ & $\alpha+\lfloor \frac{\alpha}{2} \rfloor+4$ & \textbf{closing} \\
    \hline
    $v$ & $2\alpha+1$ & $2\alpha+3$ & opened\\
    \hline
    $w$ & $2\alpha+1$ & $\alpha+\lceil \frac{\alpha}{2} \rceil+4$ & opened\\
    \hline
    $c$ & $2\alpha$ & $\alpha+3$ & closed\\
    \hline
\end{tabular*}
\medskip

\item node $v$ closes:\medskip\par
\begin{tabular*}{0.7\textwidth}{c|c|c|c}
    & \textbf{Cost if opened} & \textbf{Cost if closed} & \textbf{State after} \\
    \hline
    $x\in X$ & $2\alpha+\lfloor \frac{\alpha}{2} \rfloor+1$ & $\alpha+\lfloor \frac{\alpha}{2} \rfloor+2$ & closed \\
    \hline
    $y\in Y$ & $2\alpha+\lceil \frac{\alpha}{2} \rceil+1$ & $\alpha+\lceil \frac{\alpha}{2} \rceil+4$ & closed \\
    \hline
    $u$ & $2\alpha+1$ & $\alpha+\lfloor \frac{\alpha}{2} \rfloor+4$ & closed \\
    \hline
    $v$ & $2\alpha+\lceil \frac{\alpha}{2} \rceil+2$ & $2\alpha+3$ & \textbf{closing}\\
    \hline
    $w$ & $2\alpha+\lceil \frac{\alpha}{2} \rceil+2$ & $\alpha+2\lceil \frac{\alpha}{2} \rceil+5$ & opened\\
    \hline
    $c$ & $2\alpha+1$ & $\alpha+4$ & closed\\
    \hline
\end{tabular*}
\end{enumerate}

\caption{Calculation of improving responses in Theorem~\ref{thm:SumNotWAG}. At each time only one IR is possible, resulting in the same strategy profile after four IRs.}
\label{table:nonWagIrs}
\end{table}

\begin{figure}[t]
    \centering
%
\tikzstyle{open}=[circle,fill=blue!50,draw=black,minimum size=6pt,inner sep=0pt]
\tikzstyle{peer}=[circle,fill=white!25,draw=black,minimum size=6pt,inner sep=0pt]
\begin{tikzpicture}

\fill[green!50]
	(-2.2, 0.4) to [out=315,in=225]
	(-0.6,0.6) to [out=45,in=315]
	(-0.4,2.2) to [out=135,in=45]
	(-1.7,1.7) to [out=225,in=135]
	(-2.2,0.4);

\fill[green!50]
	(0.6,0.6) to [out=335, in=270]
	(2,1) to [out=90, in=360]
	(1,2) to [out=180,in=150]
	(0.6,0.6);

\node at (-2,1.8) {$X$};
\node at (2,1.8) {$Y$};

\node[peer] (u) at (-1.5,0) [label=below:$u$] {};
\node[peer] (v) at (0,0) [label=below:$v$] {};
\node[open] (w) at (1.5,0) [label=below:$w$] {};
\node[peer] (c) at (0,1.5) [label=above:$c$] {};
\draw (u) -- (v) -- (w);
\draw (c) -- (v);

\node[peer] (l1) at (-.75,.75){};
\node[peer] (l2) at (-2,.5){};
\node[peer] (l3) at (-.5,2){};
\node[peer] (l4) at (-1.5,1.5){};

\draw (u) -- (l1) -- (c) -- (l3) -- (l4) -- (l2) -- (u) -- (l4);
\draw (u) -- (l3) -- (l2) -- (c) -- (l4) -- (l1) -- (l2); 
\draw (l3) -- (l1);

\node[peer] (r1) at (.75,.75){};
\node[peer] (r2) at (1.75,1){};
\node[peer] (r3) at (1,1.75){};

\draw (w) -- (r1) -- (c) -- (r3) -- (r2) -- (w);
\draw (w) -- (r3) -- (r2) -- (c) -- (r1) -- (r2); 
\draw (r3) -- (r1);

%
%
%
%
%
%
%
%
%

\end{tikzpicture}
    \caption{Improving response cycle for the SUM-game with $\alpha\coloneqq 7$.
        Starting with $w$ being a gateway, nodes $u$ and $v$ change their strategies in turn and are the only nodes that want to change their strategies.}
    \label{fig:graph-not-weak-acyclic}
\end{figure}
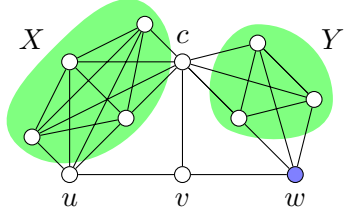

\subsection{Price of Anarchy}
In this section, we show the following result for the price of anarchy (PoA).

\begin{theorem}
\label{thm:sumPoAresults}
In the SUM-game, for $0 < \alpha < 1$ and $n(n-1) \leq \alpha$, the PoA is $1$, for $\alpha \in [1,n-1]$, it is $\LDAUTheta{n/\sqrt{\alpha}}$, and for $\alpha \in (n-1, n(n-1))$, it is $\LDAUOmicron{\sqrt{\alpha}}$.
\end{theorem}

\begin{lemma}
In the SUM-game, for $0 < \alpha < 1$ the PoA is $1$.
\end{lemma}
\begin{proof}
Let $G=(V,E)$ be a graph.
By Lemma~\ref{lemma:socialOptSmallN}, we know that $V=S$ is the social optimum.
For $\alpha < 1$, for every closed node it is an IR to open, which results in all nodes being gateways.
\end{proof}

\begin{lemma}
In the SUM-game, for $1 \leq \alpha < 2$ the PoA is $\LDAUTheta{n/\sqrt{\alpha}}$.
\end{lemma}
\begin{proof}
If $\diam(G) \geq 2$, then all nodes will open and constitute a socially optimal solution.
Otherwise we have $\diam(G) = 1$, which is a clique.
Then, the only equilibria are all nodes being gateways or having exactly one gateway.
In the latter case, the social costs are $\alpha + n(n-1)$, which yield for the considered range of $\alpha$ a PoA of $\LDAUTheta{n/\sqrt{\alpha}}$.
\end{proof}

\begin{lemma}
\label{lemma:PoAlowerBoundAlphaSmallerN}
In the SUM-game, for $2 \leq \alpha \leq n-1$ the PoA is at least $\LDAUOmega{n/\sqrt{\alpha}}$.
\end{lemma}
\begin{proof}
First, consider $\alpha\in [2,4)$ and a star graph with one center node $u$ and $n-1$ satellite nodes.
If exactly one satellite node is a gateway, this graph forms a NE with social costs $2(n-1)n$.
Comparing this to the social optimum of $\alpha n$, we get $\PoA \geq 2(n-1)/\alpha \geq (n-1)/\sqrt{\alpha}$.

Next, consider $\alpha \geq 4$.
We construct a star-like graph (cf.\ Figure~\ref{fig:sumPoALowerBound}) consisting of one center node $u$, $k \coloneqq \left\lfloor\frac{n-1}{\lfloor\sqrt{\alpha}\rfloor - 1}\right\rfloor$ many disjoint paths $p_1,\ldots,p_k$, each consisting of $\lfloor\sqrt{\alpha}\rfloor - 1$ nodes, and possibly a path $p_{k+1}$ consisting of the remaining nodes.
The first node of each path is connected to $u$.
Let one leaf node $v$ at distance exactly $\lfloor\sqrt{\alpha}\rfloor - 1$ to $u$ be a gateway.
Then, no node can perform an IR, since the maximal distance costs decrease by opening is
$\sum_{i=1}^{\lfloor\sqrt{\alpha}\rfloor - 1} 2i < \alpha$.
We estimate a social costs lower bound by considering the private costs of $u$, which are minimal for all nodes in $G$, i.e.,
$c_u(S) \geq \sum_{i=1}^{\lfloor\sqrt{\alpha}\rfloor - 1} i
    = \frac{k}{2} (\lfloor\sqrt{\alpha}\rfloor -1)\lfloor\sqrt{\alpha}\rfloor
$.
This gives for the social costs
$c(S) \geq \frac{n}{2} \left\lfloor\frac{n-1}{\lfloor\sqrt{\alpha}\rfloor - 1}\right\rfloor (\lfloor\sqrt{\alpha}\rfloor -1)\lfloor\sqrt{\alpha}\rfloor$.
Comparing this to the socially optimal cost $\alpha n$ gives $\PoA = \LDAUOmega{n / \sqrt{\alpha}}$.
\end{proof}

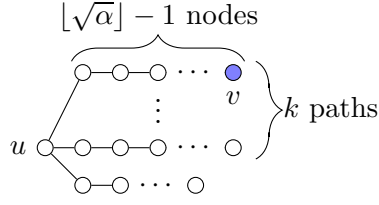
\begin{figure}[t]
    \centering
\usetikzlibrary{patterns,arrows,decorations.pathreplacing}
\tikzstyle{open}=[circle,fill=blue!50,draw=black,minimum size=6pt,inner sep=0pt]
\tikzstyle{peer}=[circle,fill=white!25,draw=black,minimum size=6pt,inner sep=0pt]
\begin{tikzpicture}

\node[peer] (u) at (0,0) [label=left:$u$] {};

\node[peer] (p11) at (.5,0) {};
\node[peer] (p12) at (1,0) {};
\node[peer] (p13) at (1.5,0) {};
\node (p14) at (2,0) {$\cdots$};
\node[peer] (p15) at (2.5,0) {};

\node[peer] (p21) at (.5,1) {};
\node[peer] (p22) at (1,1) {};
\node[peer] (p23) at (1.5,1) {};
\node (p24) at (2,1) {$\cdots$};
\node[open] (p25) at (2.5,1) [label=below:$v$] {};

\node[peer] (p31) at (.5,-.5) {};
\node[peer] (p32) at (1,-.5) {};
\node (p33) at (1.5,-.5) {$\cdots$};
\node[peer] (p34) at (2,-.5) {};

\node at (1.5,.6) {$\vdots$};

\draw (u) -- (p11) -- (p12) -- (p13) -- (p14) -- (p15);
\draw (u) -- (p21) -- (p22) -- (p23) -- (p24) -- (p25);
\draw (u) -- (p31) -- (p32) -- (p33) -- (p34);

\draw [decorate,decoration={brace,amplitude=10pt}](2.8,1.15) -- (2.8,-0.15) node[black,midway,xshift=1cm] {$k$ paths};
\draw [decorate,decoration={brace,amplitude=10pt}](.35,1.15) -- (2.65,1.15) node[black,midway,yshift=0.6cm] {$\lfloor\sqrt{\alpha}\rfloor - 1$ nodes};
\end{tikzpicture}
    \caption{NE construction for the SUM-game that gives a lower bound on the price of anarchy with $\alpha \geq 4, k \coloneqq \left\lfloor\frac{n-1}{\lfloor\sqrt{\alpha}\rfloor - 1}\right\rfloor$, and $v$ being the only gateway.}
    \label{fig:sumPoALowerBound}
\end{figure}

\begin{lemma}
In the SUM-game, for $2 \leq \alpha \leq n-1$, the PoA is $\LDAUTheta{n/\sqrt{\alpha}}$.
\end{lemma}
\begin{proof}
Let $G=(V,E)$ be a network with $|V|=n$ and $S\subseteq V$ an arbitrary NE strategy profile.
Using Lemma~\ref{lemma:socialOptSmallN}, $S = V$ forms the socially optimal solution.

Now, consider $S \not = V$.
If at least one node is a non-gateway, we get $|S| \leq \lceil\alpha\rceil$, since otherwise opening a non-gateway would reduce distance costs by more than $\alpha$.
Further, for every non-gateway $x\in V\setminus S$, we get that $d(x,S) \leq 2\lceil\sqrt{\alpha}\rceil$, since otherwise opening $x$ would reduce its private costs by at least
$\sum_{i=1}^{\lceil\sqrt{\alpha}\rceil}2i = \lceil\sqrt{\alpha}\rceil(\lceil\sqrt{\alpha}\rceil + 1) > \alpha$.
Thus, for all gateways $v\in S$, it holds $c_v(S) \leq \alpha + |V\setminus S|\cdot 2\lceil\sqrt{\alpha}\rceil$.
Since a non-gateway cannot have higher private costs than a gateway, we get
$c(S) \leq n\alpha + n\cdot|V\setminus S|\cdot 2\lceil\sqrt{\alpha}\rceil \leq n\alpha + 2 n^2 \lceil\sqrt{\alpha}\rceil$.
Comparing this to the social optimum gives
$\PoA \leq \frac{n\alpha + 2 n^2\lceil\sqrt{\alpha}\rceil}{\alpha n} \leq 1 + 2n/\lceil\sqrt{\alpha}\rceil = \LDAUOmicron{n/\sqrt{\alpha}}$.
Combining the bound with Lemma~\ref{lemma:PoAlowerBoundAlphaSmallerN}, this bound is tight.
\end{proof}

\begin{lemma}
In the SUM-game, for $n -1 < \alpha < n(n-1)$, the PoA is $\LDAUOmicron{\sqrt{\alpha}}$ and for $n(n-1)\leq \alpha$ the PoA is $1$.
\end{lemma}
\begin{proof}
First, we show that for an arbitrary strategy profile $S' \subseteq V$ it holds $c(S') > \alpha|S'| + n |V\setminus S'|$.
We define $N \coloneqq V \setminus S'$ to be the set of non-gateways.
Then, we get (note $|N|(|N|+1) < n|N|$, since $|S'| \geq 1$):
\begin{align*}
c(S') & \geq |N|(n-1) + |S'|(\alpha + |N|)
        = |N|n-|N|+n\alpha+n|N|-\alpha |N|-|N|^2\\
    & = 2|N|n-|N|(|N|+1)+\alpha(n-|N|)
        > \alpha(n-|N|)+|N|n
\end{align*}
Now, we consider a NE strategy profile $S$.
If $S = V$, then the social costs are $\alpha n$.
For the case $\alpha > n(n-1)$, no node wants to open and hence exactly one gateway exists, i.e., the social costs are $\alpha + n(n-1)$.
Since the social costs lower bound is minimized with $m=1$, we get $\PoA \leq \frac{\alpha + n(n-1)}{\alpha + (n-1)n} = \LDAUOmicron{1}$.

For $n(n-1) \geq \alpha \geq n$, let $m$ be the number of gateways in a NE $S$.
In $S$, the maximal distance from a non-gateway to a gateway is $2\sqrt{\alpha}$.
This gives for any gateway $v\in S$ that $c_v(S) \leq \alpha + (n-m)2\sqrt{\alpha}$ and for any non-gateway $w\in V\setminus S$ that $c_w(S)\leq (n-1)4\sqrt{\alpha}$.
The social costs can be upper bounded by
$c(S) \leq  m\alpha - m(n-m)2\sqrt{\alpha} + (n-m) 4 \sqrt{\alpha}(n-1)
    \leq m\alpha - m^2 2\alpha + 4\sqrt{\alpha}n(n-1)$.
The global maximum of this upper bound is at $\sqrt{\alpha}/4$, which has the value of $\alpha\sqrt{\alpha}/8 + 4\sqrt{\alpha}n(n-1)$.
Comparing this to the social costs lower bound of $\alpha + n(n-1)$ yields $\PoA = \LDAUOmicron{\sqrt{\alpha}}$.
\end{proof}

\section{The MAX-Game}
In this section, we consider the equilibria in the MAX-game.

\begin{theorem}
\label{thm:NPhardnessMax}
Given a graph $G=(V,E)$, it is NP-hard to compute an optimal set of gateways $S\subseteq V$ that minimizes the social costs for the MAX-game.
\end{theorem}
\begin{proof}
Let $X\coloneqq\{x_1,\ldots, x_m\}$ be a set of elements and $S_1,\ldots,S_n\subseteq X$ sets that form an instance of the NP-complete Set-Cover problem (cf.\ \textcite{karp1972}).
We can assume $m = 2n$ without loss of generality.
Given the Set-Cover instance, we construct an instance of the MAX-game as follows (cf.\ Figure~\ref{fig:np-hardness-construction}).
First, we create a clique $C$ of $k$ nodes and mark one of its nodes as $c$.
For every set $S_i$, we create a corresponding node $S_i$ and connect each set node to $c$.
For every element $x_i\in X$, we create a node $x_i$ and connect it to the set nodes $S_l$ with $x_i\in S_l$.
We use the parameters $\alpha \coloneqq 3$ and $k \coloneqq \alpha n = 3n$.

For now, consider that $c$ is a gateway node in the optimal solution $S^{\OPT}$ (we will prove this claim later).
We claim that then no other clique node $v\in C$, $v \not= c$ is a gateway.
For this, assume that $l$ further clique nodes are open in $S^{\OPT}$ and compute the social cost decrease by closing all clique nodes except of $c$.
If $l < k-1$, at most the distances of these $l$ nodes are increased by one each, which gives a cost decrease of $l\alpha - l > 0$.
Otherwise, the social cost decrease is at least $(k-1)\alpha - (k-1) - m - n = 2(3n-1) - 3n > 0$.
Hence, $c$ is the only node in $C\cap S^{\OPT}$.

Next, assume that there are $l$ open element nodes in $S^{\OPT}$.
If $l < m$ and if there are also set nodes open that form a set cover, by closing all element nodes only the maximal distances of these element nodes increase and the social cost decrease by at least $\alpha l - l >0$.
If there is not yet a set of set nodes open that form a set cover, we have to open at most $n$ set nodes to form a set cover.
By opening them and simultaneously closing all element nodes, the maximum distances for all clique nodes decrease by one each, which gives a social cost decrease of at least $\alpha l + k - \alpha n - l = 3 l + 3n - 3n - l > 0$.
Finally, if $l = m$, by closing all element nodes and opening a set cover, the social costs decrease by at least $\alpha m - \alpha n- m - k = 6n - 3n - 2n >0$.

Finally, we can see that $c$ actually has to be a gateway in $S^{\OPT}$.
For this, consider an arbitrary optimal setting with all clique nodes closed (if one clique node is open, we can close it and open $c$ without increasing social cost).
When opening $c$, we know that without increasing social cost we can close all element nodes and open corresponding set nodes.
Hence, when opening $c$ we can assume that all element nodes are closed and that for each element node a corresponding set is open.

Hence, the social optimal solution $S^{\OPT}$ is given by a gateway node $c$ and a minimal number of set nodes such that all element nodes are covered.
\end{proof}

\subsection{Equilibria and Convergence Properties}
Given a graph $G=(V,E)$ and $\alpha < 1$, as long as a non-gateway exists, one can always find one of them that can perform an IR by opening.
Since with $\alpha < 1$ no gateway will close, we always find a convergence sequence to a NE.
Similar for $\alpha > \diam(G)$, for a gateway, it is always an IR to close and a non-gateway will never open.
Hence, in both cases equilibria exist.

For every graph with girth of at least $4\alpha$, the following lemma also computes a NE setting.
The girth of a graph is the length of a shortest cycle in the graph, and, for acyclic graphs, it is defined to be $\infty$.
Hence, for trees always a NE exists.
\begin{lemma}
\label{lemma:maxNeExistenceMoreGeneral}
For each graph $G=(V,E)$ with $\girth(G) \geq 4\alpha$, for $\alpha \in [1,\diam(G))$ a MAX-game NE exists.
\end{lemma}
\begin{proof}
Let $x_1,x_2$ be two maximal distant nodes in $G$.
If $d(x_1,x_2) < 2\alpha$, we get with $\girth(G) \geq 4\alpha$ that $G$ is a tree and there exists a node $v$ that has maximal distance of less than $\alpha$ to every node.
In this case, opening $v$ gives a NE.

Otherwise, define $R \coloneqq \lfloor\min\{\alpha - 1, (d(x_1,x_2) - \alpha) / 2\}\rfloor$.
Since $x_1$ and $x_2$ are at maximal distance, none of them can be connected to a leaf node.
For both of these nodes, we do the following:
We consider the breadth-first-search trees $T$ up to level $R$, rooted in $x_1$ and $x_2$, respectively.
From the nodes at level $R$, we open a maximal set such that no two gateways are at distance less than $R$.

We now claim that for every node $x$ in such a tree, there exists a gateway in distance of at most $R$.
For this, consider a shortest path to a node $u$ at level $R$.
If $u$ is not a gateway, there must be another node $u'$ also at level $R$ that is a gateway.
Since the girth is at least $4\alpha$ and $R< \alpha<\diam(G)$, the shortest path from $u$ to $u'$ can only consist of nodes of the tree and hence $d(x,u') < R$.

Next, iteratively open a maximal set of further nodes such that each new node has minimal distance of exactly $\lceil\alpha\rceil$ to a gateway.
By construction, since every non-gateway has maximal distance of $\lfloor\alpha\rfloor$ to a gateway, a non-gateway can improve its maximal distance by at most $\lfloor\alpha\rfloor$ and hence cannot perform any IR.
For every gateway $v$, it holds that its private costs are $c(v) = \alpha + R$ (with both $x_1$ and $x_2$ at maximal distance, since otherwise we get a contradiction to the maximal distance of $x_1$ and $x_2$.)
Considering the private cost change of closing $v$, its maximal distance increases by exactly $\lceil\alpha\rceil$ and hence is not an IR.
\end{proof}

\begin{theorem}
\label{thm:maxNotFIPG}
For $\alpha > 1$, in general the MAX-game is not a FIPG.
\end{theorem}
\begin{proof}
Consider a graph consisting of $n \coloneqq 3 \lfloor \alpha \rfloor + 4$ nodes that are connected as a line.
We denote the first node of the line as $u$, the node at distance $\lfloor\alpha\rfloor + 1$ to $u$ as $v$, and the node at distance $2\lfloor\alpha\rfloor + 2$ to $u$ as $w$.
Initially, only $u$ is a gateway.
Then, $v$ and $w$ form an improving response cycle:
\begin{enumerate}[I:]
    \item $w$ opens since $2\lfloor \alpha \rfloor + 2 > \alpha + \lfloor \alpha \rfloor +1$.
    \item $v$ opens since $2\lfloor \alpha \rfloor + 2 > \alpha + \lfloor \alpha \rfloor + 1$.
    \item $w$ closes since $\alpha + \lfloor \alpha \rfloor + 1 > \lfloor \alpha \rfloor+1$.
    \item $v$ closes since $\alpha + 2\lfloor \alpha \rfloor + 2 > 2\lfloor \alpha \rfloor + 2$.
\end{enumerate}
Hence, the game does not provide the finite improvement property.
\end{proof}

\subsection{Price of Anarchy}
For $\alpha < 1$, as argued before, $S = V$ forms the only NE.
Since this is also the socially optimal solution, both the price of anarchy and the price of stability are $1$.
For all other ranges of $\alpha$ we show:
\begin{theorem}
\label{thm:maxPoaResults}
In the MAX-game, for $\alpha \geq 1$ the PoA is $\LDAUTheta{1 + n/\sqrt{\alpha}}$.
\end{theorem}

We now present a lower bound and then a corresponding upper bound that together yield this theorem.

\begin{lemma}
In the MAX-game, for $\alpha \geq 1$ the PoA is $\LDAUOmicron{1 + n/\sqrt{\alpha}}$.
\end{lemma}
\begin{proof}
Let $G=(V,E)$ be a graph and $S\subseteq V$ an arbitrary NE strategy profile.
For the social costs of $S$ it holds, $c(S) \leq nD$ with $D \coloneqq \diam(G)$.

We now consider the minimal social costs when placing exactly $k$ gateways on a longest shortest path $p$.
Having only these $k$ gateways, the total costs of the nodes on $p$ are
$
    \alpha k + 2k \sum_{i=1}^{\lfloor D/(2k) \rfloor} \left(i + \left\lfloor\frac{D}{2k}\right\rfloor\right) \geq \alpha k + \frac{3}{4k}D^2.
$
The total cost of nodes not on $p$ are at least $(n-D)\frac{D}{2k}$, which gives a social costs lower bound of
$\alpha k + \frac{3}{4k}D^2 + (n-D)\frac{D}{2k} = \alpha k + \frac{D^2 + 2nD}{4k}$.
This term is minimized with
$k=\sqrt{\frac{D^2+2nD}{4\alpha}}$ and then gives social costs of at least
$\sqrt{\alpha(D^2 + 2nD)}$.
We estimate an upper bound for the PoA by comparing to the NE social costs upper bound and get
$\frac{nD}{\sqrt{\alpha(D^2 + 2nD)}} \leq n / \sqrt{\alpha}$.
\end{proof}

\begin{lemma}
\label{lemma:MaxPoaLowerBound}
In the MAX-game for $\alpha \geq 1$, the PoA is $\LDAUOmega{1 + n/\sqrt{\alpha}}$.
\end{lemma}
\begin{proof}
For $n\in \mathbb{N}$, $k\coloneqq \lfloor(n-1)/3\rfloor$, consider the following graph: denote one node $c$ as the center node, connect two disjoint paths of each $k$ many nodes to $c$, and connect one path consisting of $n - 2k - 1$ nodes to $c$.
When opening the leaf node of the last connected path, we have a NE strategy profile since no node can improve its maximum distance by opening.
The social costs of this equilibrium are at least
$3\sum_{i=1}^k (i + k)
    = 3k^2 + \frac{3}{2} (k+1)k
    = \LDAUOmega{n^2}
$.

Next, consider the socially optimal solution (cf.\ previous lemma).
For $\sqrt{\alpha} \geq n$, the optimal solution coincides with the NE.
Otherwise, we get the optimal solution by opening $c$ and opening a maximal set of nodes on each path such that between each two neighboring gateways their distance is $\lfloor\alpha\rfloor$.
The social costs are at most
$\alpha \frac{n}{\lfloor\sqrt{\alpha}\rfloor} + n\frac{\alpha}{\lfloor\sqrt{\alpha}\rfloor}
    = \LDAUOmicron{n\sqrt{\alpha}}
$,
which gives the lower bound.
\end{proof}

\section{Conclusion and Future Work}
We introduced a new network model to analyze effects of network interactions that are not captured by the traditional network creation games (NCGs).
The provided PoA results emphasize that for very small or big $\alpha$ (i.e., when tending to the number of nodes), equilibria are nearly optimal solutions despite of the selfish behavior of the nodes.

In comparison to NCGs, the existence of equilibria is much harder to show.
Here, the challenge is to combine the drastically reduced strategy space with the global influences of single strategy changes.
For the SUM-game, with $\alpha \in (n-1,n(n-1))$, and the MAX-game, with graphs of girth less than $4\alpha$, computing NEs seems to be an interesting problem.

Regarding the convergence, both games do not provide the finite improvement property and remarkably, the SUM-game is not even weakly acyclic.
For the MAX-game, especially for graphs with bigger girth, due to the symmetry of the maximum, the equilibria and convergence properties seem to be more stable than for the sum.

\begin{sloppy}
\printbibliography
\end{sloppy}

\end{document}